\documentclass[graybox]{svmult}

\usepackage{type1cm}        
\usepackage{makeidx}         
\usepackage{graphicx}        
\usepackage{multicol}        
\usepackage[bottom]{footmisc}

\usepackage{newtxtext}       %
\usepackage{newtxmath}       
\usepackage{enumitem}

\usepackage{bbm}

\def\bR{\mathbb{R}}
\def\bC{\mathbb{C}}
\def\bN{\mathbb{N}}
\def\bZ{\mathbb{Z}}

\def\cA{\mathcal{A}}
\def\cF{\mathcal{F}}
\def\cS{\mathcal{S}}

\def\cB{\mathcal{B}}

\def\cM{\mathcal{M}}

\def\rd{\bR^d}

\def\rdd{\bR^{2d}}

\def\la{\langle}
\def\ra{\rangle}
\def\lc{\left(}
\def\rc{\right)}

\def\supp{\mathrm{supp}}

\def\*b{*_{\bullet}}
\def\w{\mathrm{w}}

\def\S0{S^0_{0,0}}

\def\Bd'{B_{\delta'}}

\def\cBd'{\bar{B}_{\delta'}}

\def\hbi{\frac{i}{\hbar}}
\def\Spdr{\mathrm{Sp}(d,\bR)}

\def\te{\widetilde{E}}
\def\rte{\bR \setminus \te}

\makeindex             

\begin{document}

\title*{On exceptional times for pointwise convergence of integral kernels in Feynman-Trotter path integrals}
\titlerunning{Exceptional times for pointwise convergence of integral kernels in path integrals}
\author{Hans G. Feichtinger, Fabio Nicola, and S. Ivan Trapasso}
\institute{Hans G. Feichtinger \at Faculty of Mathematics, University of Vienna, Oskar-Morgenstern-Platz 1, A-1090 Wien. \newline \email{hans.feichtinger@univie.ac.at} 
\and Fabio Nicola \at Dipartimento di Scienze Matematiche ``G. L. Lagrange'', Politecnico di Torino,
\newline Corso Duca degli Abruzzi 24, 10129, Torino. \email{fabio.nicola@polito.it} \and S. Ivan Trapasso \at Dipartimento di Scienze Matematiche ``G. L. Lagrange'', Politecnico di Torino, \newline  Corso Duca degli Abruzzi 24, 10129, Torino. \email {salvatore.trapasso@polito.it}}
%
%
\maketitle

\abstract{In the first part of the paper we provide a survey of recent results concerning the problem of pointwise convergence of integral kernels in Feynman path integral, obtained by means of time-frequency analysis techniques. We then focus on exceptional times, where the previous results do not hold, and we show that weaker forms of convergence still occur. In conclusion we offer some clues about possible physical interpretation of exceptional times.}

\section{Introduction}
Integration over infinite-dimensional spaces of paths plays a relevant role in modern quantum physics. This machinery first appeared in a 1948 paper \cite{feyn1 48} by Richard Feynman, shortly followed by \cite{feyn 49} where path integrals paved the way to the celebrated Feynman diagrams, hence to a completely new way to investigate field theories.

Let us briefly recall the most important features of the functional integral formulation of (non-relativistic) quantum mechanics. The interested reader may consult the textbook \cite{feyn2 hibbs} for a comprehensive introduction to the subject from a physical perspective. Recall that the state of a particle in $\rd$ at time $t\in \bR$ is represented by the wave function $\psi(t,x)$, $(t,x) \in \bR \times \rd$, such that $\psi (t,\cdot)\in L^{2}(\rd)$. The time evolution of the initial state $\varphi(x)$ at $t=0$ is regulated by the corresponding Cauchy problem for the Schr\"odinger equation:
\begin{equation}
\begin{cases}
i \hbar \partial_{t}\psi=(H_{0}+V(x))\psi\\
\psi(0,x)=\varphi(x),
\end{cases}\label{cauchy schr}
\end{equation}
where $0<\hbar\le 1$ is a parameter (representing the Planck constant), $H_{0}=-\hbar^2\triangle/2$ is the free particle Hamiltonian and $V$ is a real-valued potential; we set $m=1$ for the mass of the particle.
The map $U(t,s): \psi(s,\cdot) \mapsto \psi(t,\cdot)$, $t,s\in\bR$, is a unitary operator on $L^2(\rd)$ and is called \textit{propagator} or \textit{evolution operator}; we set $U(t)$ for $U(t,0)$. Since $U(t)$ is a linear operator we may formally represent it as an integral operator with distribution kernel $u_t$, namely
\[
\psi(t,x)=\int_{\rd}u_t(x,y)\varphi(y)dy.
\]
The kernel $u_t$ (actually known as propagator in physics) is interpreted as the transition amplitude from the position $y$ at time $0$ to the position $x$ at time $t$. In his papers Feynman essentially provided a recipe for how to compute this kernel, involving all the possible \emph{interfering alternative paths} from $y$ to $x$ that could be followed by the particle. In particular, each path would contribute to the total probability amplitude with a phase factor proportional to the \emph{action functional} corresponding
to the path:
\[
S\left[\gamma\right] = S(t,0,x,y)=\int_{s}^{t}L(\gamma(\tau),\dot{\gamma}(\tau))d\tau,
\]
where $L$ is the Lagrangian of the corresponding classical system. In a nutshell, a formal representation of the kernel is
\begin{equation}
u_t(x,y)=\int e^{\hbi S\left[\gamma\right]}\mathcal{D}\gamma,\label{ker path int}
\end{equation}
underpinning some integration procedure over the infinite-dimensional space of paths satisfying the conditions above. Notice that (a still formal) application of the stationary phase principle shows that the semiclassical limit $\hbar \to 0$ selects the classical trajectory, in according with the principle of stationary action of classical mechanics.

\subsection{The mathematics of path integrals}
In spite of the popularity and the successful predictions of path integrals, it is not clear what the meaning of \eqref{ker path int} could be from a mathematical point of view. This is in fact an open subfield of functional analysis and there have been several attempts to provide a rigorous and satisfactory theory of path integrals with the support of techniques ranging from infinite-dimensional analysis to operator theory, but also from stochastics to geometry. We cannot hope to frame here more than seventy years of literature; we suggest the monographs \cite{albeverio book,fujiwara5 book,kl11,mazzucchi} as points of departure as well as the article \cite{albeverio schol} for a broad overview. We remark that only in recent times techniques from time-frequency analysis have been fruitfully used in the study of mathematical path integrals, see for instance \cite{nicola1 conv lp,nt,nt pointw}; see also \cite{T20} for an expository paper on the topic.

Among the several frameworks mentioned above we focus here on the so-called \textit{sequential approach}, introduced by Nelson in \cite{nelson}. The reasons behind this choice are manifold; first, it is probably the mathematical scheme which best meets Feynman's original insight and some of its features are nowadays part of the custom in physics literature, cf. \cite{gs,kleinert}. Moreover, the perturbative nature of this approach is very well suited to certain function spaces and operators related to time-frequency analysis, as will be elucidated later.

Nelson's approach relies on two issues. Recall that the evolution operator for the Schr\"odinger equation with $V=0$, namely $U_0(t)=e^{-\hbi tH_{0}}$, $H_0=-\hbar^2\triangle/2$, is a Fourier multiplier; an explicit representation can be derived after standard computation (cf. \cite[Sec. IX.7]{reed simon 2}):
\begin{equation} \label{free prop int}
e^{-\hbi tH_{0}}\varphi\left(x\right)=\frac{1}{\left(2\pi it\hbar \right)^{d/2}}\int_{\mathbb{R}^{d}}\exp\left(\hbi\frac{\left|x-y\right|^{2}}{2t}\right)\varphi(y)dy,\qquad\varphi\in\mathcal{S}(\mathbb{R}^{d}).
\end{equation}

The second ingredient is a well-known tool from the theory of operator semigroups. Under suitable conditions on the domain of $H_0$ and on the potential $V$\footnote{For instance one may consider a potential $V$ such that $H_0+V$ is essentially self-adjoint on $D(H_0)\cap D(V)$, cf. \cite[Sec. VIII.8]{reed simon 1}.}, the \textit{Trotter product formula} holds for the semigroup generated by $H=H_{0}+V$, namely it can be expressed in terms of a strong operator limit in $L^2(\rd)$:
\[
e^{-\hbi t(H_{0}+V)}=\lim_{n\rightarrow\infty}\left(e^{-\frac{i}{\hbar}  \frac{t}{n}H_{0}}e^{-\hbi\frac{t}{n}V}\right)^{n}.
\]
The joint application of these two results gives that the complete propagator $e^{-\hbi tH}$ can be expressed as limit of integral operators (cf.\ \cite[Thm.\ X.66]{reed simon 2}):
\begin{equation}
e^{-\hbi t\left(H_{0}+V\right)}\varphi(x)=\lim_{n\rightarrow\infty}\left(2\pi \hbar i\frac{t}{n}\right)^{-\frac{nd}{2}}\int_{\mathbb{R}^{nd}}e^{ \hbi S_{n}\left(t;x_{0},\ldots,x_{n-1},x\right)}\varphi\left(x_{0}\right)dx_{0}\ldots dx_{n-1},\label{prop ts limit}
\end{equation}
where we set
\[
S_{n}\left(t;x_{0},\ldots,x_{n-1},x\right)=\sum_{k=1}^{n}\frac{t}{n}\left[\frac{1}{2}\left(\frac{\left|x_{k}-x_{k-1}\right|}{t/n}\right)^{2}-V\left(x_{k}\right)\right], \quad x_0=y, \,x_{n}=x.
\]
The role of the phase $S_{n}\left(t;x_{0},\ldots,x_{n}\right)$ may be clarified by the following argument: given the points $x_{0},\ldots,x_{n-1},x\in\mathbb{R}^{d}$,
let $\overline{\gamma}$ be
the polygonal path through the vertices $x_{k}=\overline{\gamma}\left(kt/n\right)$,
$k=0,\ldots,n$, $x_n=x$, parametrized as
\begin{equation}\label{broken}
\overline{\gamma}\left(\tau\right)=x_{k}+\frac{x_{k+1}-x_{k}}{t/n}\left(\tau-k\frac{t}{n}\right),\qquad\tau\in\left[k\frac{t}{n},\left(k+1\right)\frac{t}{n}\right],\qquad k=0,\ldots,n-1.
\end{equation}
Hence $\overline{\gamma}$ prescribes a classical motion with constant velocity along each segment. The action functional for such path is given by
\[
S\left[\overline{\gamma}\right]=\sum_{k=1}^{n}\frac{1}{2}\frac{t}{n}\left(\frac{\left|x_{k}-x_{k-1}\right|}{t/n}\right)^{2}-\int_{0}^{t} V(\overline{\gamma}(\tau))d\tau.
\]
According to Feynman's heuristics, the relation in \eqref{prop ts limit} should be interpreted as the definition of an integral over all polygonal paths while $S_{n}\left(x_{0},\ldots,x_{n},t\right)$
is a Riemann-like, finite-dimensional approximation of the action functional evaluated on them. The regime $n\rightarrow\infty$ is then intuitively clear: the set of polygonal paths becomes the set of all paths and we recover \eqref{ker path int}.

\subsection{Convergence at the level of integral kernels}
The sequential approach discussed above seems to suggest that path integral can be made mathematically rigorous at the level of operators rather than integral kernels. This remark is reinforced by the achievements of different mathematical theories of path integrals relying on the standard operator-theoretic approach to quantum mechanics. Consider for instance the so-called \textit{time slicing approximation approach} introduced by Fujiwara in celebrated papers like \cite{fujiwara1 fund sol,fujiwara2 duke} - see also the monograph \cite{fujiwara5 book} for a systematic exposition; broadly speaking, the philosophy underlying these works is to design sequences of finite-dimensional approximation operators on $L^2(\rd)$ (in particular, oscillatory integral operators) and then prove convergence to the exact propagator $U(t)$ in some operator topology on $L^2$.

Actually, there are good reasons for not being completely satisfied with this state of affairs. The lesson of Feynman's original formulation strongly motivates a focus shift from operators to their kernels, in particular to the problem of \textit{pointwise} convergence of the integral kernels in \eqref{prop ts limit} to the kernel $u_t$ of the propagator. This may appear as an unaffordable problem in general since non-regular or even purely distribution kernels may show up, thus the problem of convergence can be hard or even pointless. A strong clue pointing in this direction comes from the already mentioned papers by Fujiwara, where convergence in a finer topology at the level of integral kernels is proved for sufficiently small time intervals and smooth potentials with at most quadratic growth.

We describe below the recent results obtained by two of the authors in \cite{nt}, where techniques of time-frequency analysis are fruitfully used to prove pointwise convergence of integral kernels in the framework provided by the sequential approach. In contrast with the aforementioned results by Fujiwara we consider bounded   potentials (the minimal regularity assumption is continuity) and we obtain the desired convergence for the kernels in suitable topologies which imply pointwise convergence. Our results are global in time, namely they hold for any fixed $t\in \rte$, where $\te$ is a set of \textit{exceptional times}. We describe below the most important features of this set from both the mathematical and physical points of view and provide explicit examples. For the moment we confine ourselves to remark that exceptional times are to be expected: recall that the involved kernels are in general tempered distributions in $\cS'(\rd)$ in view of the Schwartz kernel theorem and the problem of pointwise convergence is well-posed only when the kernels are actually functions. One may still wonder whether there is convergence at exceptional times in some weaker distribution sense. We are able to prove global-in-time convergence in this fashion, again supported by the framework of time-frequency analysis techniques and function spaces. In order to precisely state and prove the claimed results we devote the next section to collect some preparatory material.

\section{Preliminaries} \label{sec prel}
\subsection{Notation} We set $x^2=x\cdot x$, for $x\in\mathbb{R}^d$, where $x\cdot y$ is the scalar product on $\mathbb{R}^{d}$. The Schwartz class is denoted by  $\mathcal{S}(\mathbb{R}^{d})$, the space of tempered distributions by  $\mathcal{S}'(\mathbb{R}^{d})$. The brackets  $\langle  f,g\rangle $ denote the extension to $\cS' (\mathbb{R}^{d})\times \cS (\mathbb{R}^{d})$ of the inner product $\langle f,g\rangle=\int_{\rd} f(x){\overline {g(x)}}dx$ on $L^2(\mathbb{R}^{d})$, but also other related dualities described below.

The conjugate exponent $p'$ of $p \in [1,\infty]$ is defined by $1/p+1/p'=1$. The symbol $\lesssim$ means that the underlying inequality holds up to a positive constant factor $C>0$.
For any $x\in\mathbb{R}^{d}$ and $s\in\mathbb{R}$
we set $\left\langle x\right\rangle ^{s}\coloneqq(1+\left|x\right|^{2})^{s/2}$.
We choose the following normalization for the Fourier transform:
\[
\mathcal{F}f\left(\xi\right)=\int_{\mathbb{R}^{d}}e^{-2\pi ix \cdot \xi}f(x) d x,\qquad\xi\in \rd. \]
We define the translation and modulation operators: for
any $x,\xi \in\mathbb{R}^{d}$ and $f\in\mathcal{S}(\mathbb{R}^{d})$,
\[
\left(T_{x}f\right)\left(y\right)\coloneqq f(y-x),\qquad\left(M_{\xi}f\right)(y)\coloneqq e^{2 \pi i \xi \cdot y}f(y).
\]
These operators can be extended by duality on tempered distributions. The composition $\pi(x,\xi)=M_\xi T_x$ constitutes a so-called \textit{time-frequency shift}.

Given a normed linear space of distributions $X \subset \cS'(\rd)$, we set
\[ X_{\mathrm{comp}} \coloneqq \{u \in X \,:\, \supp(u) \text{ is a compact subset of }\rd \},  \]
\[ X_{\mathrm{loc}} \coloneqq \{u \in \cS'(\rd) \,:\, \phi u \in X \,\, \forall \phi \in C^{\infty}_c(\rd) \}. \]

In the rest of the paper we set $\hbar=1$ for convenience, since we are not concerned with semiclassical aspects.

\subsection{Modulation spaces}

The short-time Fourier transform (STFT) of a tempered distribution $f\in\cS'(\rd)$ with respect to the window function $g \in \cS(\rd)\setminus\{0\}$ is defined by
\begin{equation}\label{FTdef}
V_gf (x,\xi)\coloneqq \langle f, \pi(x,\xi) g\rangle=\cF (f\cdot T_x g)(\xi)=\int_{\rd}e^{-2\pi iy \cdot \xi }
f(y)\, {\overline {g(y-x)}}\,dy.
\end{equation}

The monograph \cite{gro1 book} contains a comprehensive treatment of the mathematical properties of this time-frequency representation, especially those mentioned below. We stress that the STFT is deeply connected with other well-known phase-space transforms, in particular the Wigner distribution
\begin{equation} W(f,g)(x,\xi )=\int_{\mathbb{R}^{d}}e^{-2\pi iy\cdot \xi }f\left(x+\frac{y}{2}\right)%
\overline{g\left(x-\frac{y}{2}\right)}\ dy. \label{def wig} \end{equation}

Given a non-zero window $g\in\cS(\rd)$, $s\in \bR$ and $1\leq p,q\leq \infty$, the {\it modulation space} $M^{p,q}_s(\rd)$ consists of all tempered distributions $f\in\cS'(\rd)$ such that $V_gf\in L^{p,q}_s(\rdd)$ (mixed weighted Lebesgue space), that is:
\[ \|f\|_{M^{p,q}_s}=\|V_gf\|_{L^{p,q}_s}=\left(\int_{\rd}
\left(\int_{\rd}|V_gf(x,\xi)|^p\,
dx\right)^{q/p} \la \xi \ra^{qs} d\xi \right)^{1/q}  \, <\infty ,\]
with trivial modification if $p$ or $q$ is $\infty$.
If $p=q$, we write $M^p$ instead of $M^{p,p}$, while for the unweighted case ($s=0$) we set $M_{0}^{p,q}\equiv M^{p,q}$.

It can be proved that $M^{p,q}_s(\rd)$ is a Banach space whose definition does not depend on the choice of the window $g$. We mention that many common function spaces are intimately related with modulation spaces: for instance,
\begin{enumerate}[label=(\roman*)]
	\item $M^2(\rd)$ coincides with the Hilbert space $L^2(\rd)$;
	\item $M^2_s(\rd)$ coincides with the usual $L^2$-based Sobolev space $H^s(\rd)$;
	\item the following continuous embeddings with Lebesgue spaces hold: \[ M^{p,q}_r(\rd) \hookrightarrow L^p(\rd) \hookrightarrow M^{p,q}_s(\rd), \qquad r>d/q' \text{ and } s<-d/q. \] In particular, \[ M^{p,1}(\rd) \hookrightarrow L^p(\rd) \hookrightarrow M^{p,\infty}(\rd).\]
\end{enumerate}
For these and other embeddings we address the reader to \cite{fei new segal,fei modulation 83,fe03-1,gro1 book}.

We wish to focus on distinguished members of the family of modulation spaces. The Banach-Gelfand triple $(M^1(\rd),L^2(\rd),M^{\infty}(\rd))$ proved to be a very fruitful generalization of the standard triple $(\cS(\rd),L^2(\rd),\cS'(\rd))$ for the purposes of time-frequency analysis, see \cite{CFL,fei Z,jakob} for further details. The space $M^1(\rd)$ is also known as the \textit{Feichtinger algebra} \cite{fei new segal} and it does enjoy a large number of particularly nice properties. We stress that $\cS(\rd)\subset M^1(\rd)$ and $L^{2}(\rd)$ is the completion of $M^1(\rd)$ with respect to $\left\Vert \cdot\right\Vert _{L^{2}}$ norm. Moreover $(M^1(\rd))' = M^{\infty}(\rd)$ under the duality
\[ \la f,\phi \ra = \int_{\rdd} V_gf(z)\overline{V_g\phi}(z)dz,\quad  f \in M^1(\rd), \,\, \phi \in M^{\infty}(\rd), \] for any $g \in \cS(\rd)\setminus\{0\}$, without loss of generality
with $\|g\|_2 = 1$. 
Finally, $M^1(\rd)$ is isometrically invariant under Fourier transform and arbitrary time-frequency shifts, and the embedding $M^1(\rd)\hookrightarrow M^{p,q}(\rd)$ hold for all $1 \le p,q \le \infty$.
An additional benefit of this extended framework is that one may derive a streamlined and self-consistent presentation of the mathematical foundations of signal analysis with a limited amount of technicalities, cf. \cite{feja20}.

The role of $(M^1,L^2,M^\infty)$ as a Gelfand triple is further reinforced by the \textit{Feichtinger kernel theorem} \cite{fe80,fegr92-1,fei gro kernel,cn ker}.

\begin{theorem} \label{fei ker thm} \begin{enumerate}[label=(\roman*)]
		\item Every distribution $k \in M^{\infty}(\mathbb{R}^{2d})$ defines a bounded linear operator $T : M^1(\mathbb{R}^d) \rightarrow M^{\infty}(\mathbb{R}^d)$ according to
		\[ \langle T f,g \rangle = \langle k, g \otimes \overline{f} \rangle, \quad \forall f,g\in M^1(\mathbb{R}^d),\]
		with $\left\Vert T \right\Vert_{M^1\rightarrow M^\infty} \lesssim \left\Vert k \right\Vert_{M^{\infty}}$.
		\item Any linear bounded operator $T : M^1(\mathbb{R}^d) \rightarrow M^{\infty}(\mathbb{R}^d)$ arises in this way for a unique kernel $k \in M^{\infty}(\mathbb{R}^{2d})$; moreover  $\left\Vert k \right\Vert_{M^{\infty}}\lesssim \left\Vert T \right\Vert_{M^1\rightarrow M^\infty}$.
	\end{enumerate}
\end{theorem}

Another interesting modulation space is $M^{\infty,1}(\rd)$, also known as the \textit{Sj\"ostrand class} since it was highlighted in \cite{sjo} as an exotic symbol class still yielding bounded pseudodifferential operators on $L^2(\rd)$ (see the next section for further details, also \cite{gro2 sj,grst07}). In order to specify the regularity of functions in this space recall the definition of the Fourier-Lebesgue space: for $s\in \bR$ we set
\[
f\in\cF L_{s}^{1}(\mathbb{R}^{d})\quad\Leftrightarrow\quad\left\Vert f\right\Vert _{\mathcal{F}L_{s}^{1}}=\int_{\mathbb{R}^{d}}\left|\mathcal{F}f\left(\xi\right)\right| \langle \xi \rangle^s d\xi<\infty.
\]
\begin{proposition}[{\cite{gro1 book} and \cite[Prop. 3.4]{nt pointw}}] \label{prop pot sp}
	\begin{enumerate}
		\item $M^{\infty,1}(\rd) \subset (\mathcal{F}L^1)_{\rm loc}(\rd)\cap L^\infty(\rd)\subset C^0(\rd)\cap L^\infty(\rd)$.
		\item $(M^{\infty,1})_{\rm loc}(\rd) = (\mathcal{F}L^1)_{\rm loc}(\rd)=(\cF\cM)_{\rm loc}(\rd)$, where $\cF\cM(\rd)$ is the space of Fourier transforms of (finite) complex measures on $\rd$.
\item $\cF\cM(\rd) \subset M^{\infty,1}(\rd)$.
	\end{enumerate}
\end{proposition}

The equality $(\mathcal{F}L^1)_{\rm loc}(\rd)=(\cF\cM)_{\rm loc}(\rd)$ is an immediate consequence of the fact that $L^1(\rd)$ is an ideal in the convolution algebra $\cM(\rd)$.

Moreover, $M^{\infty,1}(\rd)$ is a Banach algebra under pointwise product. In fact, precise conditions are known on $p$, $q$ and $s$ in order for $M_{s}^{p,q}$ to be a Banach algebra with respect to pointwise multiplication.
\begin{proposition}[{\cite[Thm. 3.5 and Cor. 2.10]{rs mod}}] \label{Mpqs ban alg}Let $1\le p,q\le\infty$ and $s\in\mathbb{R}$. The following facts are equivalent.\\
	$(i)$ $M_{s}^{p,q}(\mathbb{R}^{d})$ is a Banach algebra for pointwise multiplication\footnote{To be precise, we provide conditions under which the embedding $M_{s}^{p,q}\cdot M_{s}^{p,q}\hookrightarrow M_{s}^{p,q}$ is continuous; this means that the algebra property holds up to a constant. It is a rather standard result that that there exists an equivalent norm for which the previous estimate holds with $C=1$ (cf. \cite[Thm.\ 10.2]{rudin fa}). This setting will be tacitly assumed whenever concerned with Banach algebras from now on.}. \\ $(ii)$ $M^{p,q}_s(\rd) \hookrightarrow L^{\infty}(\rd)$. \\ $(iii)$ Either $s=0$ and $q=1$ or $s>d/q'$.
\end{proposition}
We deduce that also the modulation spaces $M_{s}^{\infty}(\mathbb{R}^{d})$
with $s>d$ are Banach algebras for pointwise multiplication. In particular we have $M^{\infty}_s(\rd) \hookrightarrow M^{\infty,1}(\rd)$ for $s>d$ and the following characterization holds:
\begin{equation}\label{S000 char}
C^{\infty}_b(\mathbb{R}^{d})\coloneqq\left\{ f\in C^{\infty}(\mathbb{R}^{d})\,:\,\left|\partial^{\alpha}f\right|\leq C_{\alpha}\quad\forall\alpha\in\mathbb{N}^{d}\right\}=\bigcap_{s\ge0}M_{s}^{\infty}(\mathbb{R}^{d});\end{equation}
see \cite[Lemma 6.1]{gro3 rze} for further details.

\subsection{Weyl operators} \label{weyl sec}
The success of time-frequency analysis in the theory of pseudodifferential operators mainly relies on the following equality:
\begin{equation}
\langle \sigma^{\mathrm{w}}f,g\rangle =\langle \sigma,W(g,f)\rangle ,\qquad\forall f,g\in\mathcal{S}(\mathbb{R}^{d}),\label{def wig dual}
\end{equation}
where $\sigma\in\mathcal{S}'(\mathbb{R}^{2d})$ is the \textit{symbol} of the \textit{Weyl operator} $\sigma^{\mathrm{w}}:\mathcal{S}(\mathbb{R}^{d})\rightarrow\mathcal{S}'(\mathbb{R}^{d})$, which can be formally represented as
\[
\sigma^{\text{w}}f\left(x\right)\coloneqq\int_{\mathbb{R}^{2d}}e^{2\pi  i\left(x-y\right)\cdot \xi}\sigma\left(\frac{x+y}{2},\xi\right)f(y)dyd\xi,
\]
while $W(g,f)$ is the Wigner transform defined in \eqref{def wig}. The main benefit of a time-frequency approach to Weyl operators is that very general symbol classes may be taken into account, in particular modulation spaces - recall that classical symbol classes are usually defined by means of decay/smoothness conditions, such as the H\"ordmander classes $S_{\rho,\delta}^{m}(\mathbb{R}^{2d})$ \cite{hormander2 book 3}. Moreover, most of the properties of $\sigma^\w$ are intimately connected to those of the Wigner transform, the latter being very well established nowadays \cite{dG symp met,gro1 book}.

The composition of Weyl transforms induces a bilinear form on symbols, the so-called \emph{twisted product}: this means that the composition of two operators  $\sigma^{\mathrm{w}} \circ \rho^{\mathrm{w}}$ is in fact a Weyl operator with special symbol denoted by $\sigma \# \rho$. Explicit formulas for $\sigma \# \rho$ are known (cf. \cite{wong}) but we are more interested in the algebra structure induced on symbol spaces. It is indeed a peculiar feature of $M^{\infty,1}(\mathbb{R}^{2d})$, as well as of $M_{s}^{\infty}(\mathbb{R}^{2d})$ with $s>2d$, to enjoy a double Banach algebra structure: \begin{itemize}
	\item a commutative one with respect to the pointwise multiplication as a consequence of Proposition \ref{Mpqs ban alg};
	\item a non-commutative one with respect to the twisted product of symbols (\cite{gro3 rze,sjo}); for instance, $\sigma,\rho \in M^{\infty,1}(\rdd) \Longrightarrow \sigma \# \rho \in M^{\infty,1}(\rdd)$.
\end{itemize}
Furthermore, it turns out that the latter algebraic structure can be related to a characterizing sparse behaviour satisfied by pseudodifferential operators with symbols in those spaces, the so-called \textit{almost diagonalization property} with respect to time-frequency shifts; it can be proved that $\sigma\in M^{\infty}_s(\rdd)$ if and only if, for some (hence any) $g \in \cS(\rd)\setminus\{0\}$,
\[ |\la \sigma^w \pi(z)g,\pi(w)g\ra|  \le C\langle w-z \rangle^{-s}, \quad z,w\in \rdd. \]
In a similar fashion, $\sigma\in M^{\infty,1}(\rdd) $ if and only if there exists $H\in L^1(\rdd)$ such that
\[ |\la \sigma^w \pi(z)g,\pi(w)g\ra|  \le H(w-z), \quad z,w\in \rdd. \]
The reader may consult \cite{CGNR fio, CGNR jmp, CNR sparsity, CNT 18, gro2 sj, gro3 rze} for further details on this topic.

\section{Pointwise convergence of integral kernels}\label{sec pw conv}
The main results in \cite{nt pointw} require us to consider a slightly generalized version of the free Hamiltonian operator $H_0$ in \eqref{cauchy schr}. Let $a$ be a quadratic homogeneous polynomial on $\rdd$, namely
\[ a(x,\xi) = \frac{1}{2}x\cdot Ax + \xi \cdot B x + \frac{1}{2}\xi \cdot C \xi, \] for some symmetric matrices $A,C\in \bR^{d\times d}$ and $B\in \bR^{d\times d}$. The solution of \eqref{cauchy schr} with $H_0 = a^{\mathrm{w}}$ (the Weyl transform of $a$) and $V=0$ is given by
\[ \psi(t,x)=e^{-it H_0}\varphi(x)=\mu(\cA_t)\varphi(x), \]
where $\mu(\cA_t)$ is a \textit{metaplectic operator} - see \cite[Sec.\ 15.1.3]{dG symp met} and also \cite{CN pot mod, folland} for a complete derivation of this classic result.  A precise characterization of metaplectic operators would lead us too far, hence we just outline their main features. First, recall that the phase-space flow governed by the Hamilton equations\footnote{The factor $2\pi$ is a consequence of the normalization of the Fourier transform adopted in the paper.} \[2\pi \dot{z} = J \nabla_z a(z) = \mathbb{A}, \quad \mathbb{A}= \left(\begin{array}{cc} B & C \\ -A & -B^{\top}\end{array}\right), \]
defines a mapping \begin{equation}\label{block At} \bR \ni t  \mapsto \cA_t = e^{(t/2\pi)\mathbb{A}}= \left(\begin{array}{cc}
A_{t} & B_{t}\\
C_{t} & D_{t}
\end{array}\right) \in \mathrm{Sp}(d,\bR). \end{equation} In sloppy terms, any symplectic matrix $S \in \Spdr$ is associated  with a unitary bounded operator $\mu(S)$ on $L^2(\rd)$ which satisfies the intertwining property
\[ \mu(S)^{-1} \sigma^\w \mu(S) = (\sigma\circ S)^\w, \quad \sigma \in \cS'(\rdd). \] In particular, the classical flow $\cA_t$ is associated (up to a complex phase factor) with a family of unitary operators on $L^2(\rd)$ (for details see \cite{gro1 book}, Thm. 9.4.2) 
An explicit formula for $\mu(\cA_t)$ may be provided in some special cases: for all $t\in \bR$ such that $\cA_t$ is a \textit{free symplectic matrix}, namely such that the upper-right block $B_t$ is invertible, the corresponding metaplectic operator may be represented as a \textit{quadratic Fourier transform} \cite[Sec.\  7.2.2]{dG symp met}, namely
\begin{equation} \label{met int formula} \mu(\cA_t)\varphi(x) = c_t \lvert \det B_t\rvert^{-1/2} \int_{\rd} e^{2\pi i \Phi_t}(x,\xi) \varphi(y) dy, \qquad \varphi \in \cS(\rd), \end{equation} for suitable $c_t\in \bC$, $|c_t|=1$, where
\begin{equation}\label{phit}
\Phi_{t}\left(x,y\right)=\frac{1}{2}x \cdot D_{t}B_{t}^{-1}x-y\cdot B_{t}^{-1}x+\frac{1}{2}y\cdot B_{t}^{-1}A_{t}y,\qquad x,y\in\mathbb{R}^{d}.
\end{equation}
This representation of $\mu(\cA_t)$ is a main ingredient of our results, hence we stress that it does hold for any $t\in \rte$, where we define the set of \textit{exceptional times} as
\begin{equation}\label{def te} \te = \{ t \in \bR \,:\, \det B_t =0 \}. \end{equation}
Some of the properties of this set can be immediately deduced from the fact that it is indeed the zero set of an analytic function: apart from the case $\te=\bR$ (which trivially happens when $H=0$), $\te$ is a discrete (hence at most countable) subset of $\bR$ which always includes $t=0$ - in particular $\te=\{0\}$ in the case of the free Schr\"odinger equation ($V=0$).

 We now apply a version of Trotter formula from the theory of operator semigroups. It is known that $H_{0}=a^{\rm w}$ is a self-adjoint operator on the maximal domain (see \cite{hormander1 mehler})
\[
D\left(H_{0}\right)=\{ \psi \in L^{2}(\mathbb{R}^{d})\,:\,H_{0}\psi \in L^{2}(\mathbb{R}^{d})\} .
\]
For our purposes it is enough to assume that $V \in \cB(L^2(\rd))$, hence we consider bounded perturbations of $H_0$. In particular, $V \in L^{\infty}(\rd)$ is a suitable choice, even for possibly complex-valued potentials.

Then, we have (cf. for instance \cite[Cor. 2.7 and Ex. 2.9]{engel})
\begin{equation}\label{trotter maint}
e^{-it\left(H_{0}+V\right)}=\lim_{n\rightarrow\infty}E_{n}(t),\qquad E_{n}(t)=\Big(e^{-i\frac{t}{n}H_{0}}e^{-i\frac{t}{n}V}\Big)^{n},
\end{equation}
with convergence in the strong operator topology in $L^{2}(\mathbb{R}^{d})$. We denote by $e_{n,t}(x,y)$ the distribution kernel of $E_{n}(t)$ and by $u_t(x,y)$ that of $U(t)=e^{-it\left(H_{0}+V\right)}$.

We assume $V \in L^{\infty}(\rd)$, and we tune its regularity as follows. In view of the discussion on modulation spaces in the previous section, we have available a scale of decreasing regularity spaces.
\begin{enumerate} \item The best option for our purposes is given by $C^{\infty}_b(\rd)$, the space of smooth bounded functions with bounded derivatives of any order.
	\item Subsequently we have the (scale of) modulation spaces $M^{\infty}_s(\rd)$, $s>d$, which contain bounded continuous functions becoming less regular as $s\searrow d$ - the parameter $s$ can be thought of as a measure of (fractional) differentiability.
	\item We finally consider the Sj\"ostrand class $M^{\infty,1}(\rd)$, where even the partial regularity of the previous level is lost. We are still dealing with bounded continuous functions, which locally enjoy the mild regularity of the Fourier transform of a $L^1$ function.
\end{enumerate}

Let us first state our main result at the intermediate regularity encoded by $M^{\infty}_s(\rd)$.

\begin{theorem}[{\cite[Thm. 1.1]{nt pointw}}]
	\label{maint minfty}Let  $H_{0}=a^{\mathrm{w}}$
	as above and $V\in M_{s}^{\infty}(\mathbb{R}^{d})$,
	with $s>2d$. Let $\cA_t$ be the classical flow associated with $H_0$ as in \eqref{block At}. For any $t\in \rte$:
	\begin{enumerate}
		\item the distributions $e^{-2\pi i\Phi_{t}}e_{n,t}$, $n\geq 1$,
		and $e^{-2\pi i\Phi_{t}}u_t$
		belong to a bounded subset of $M_{s}^{\infty}(\mathbb{R}^{2d})$;
		\item $e_{n,t}\rightarrow u_t$ in $(\mathcal{F}L_{r}^{1})_{\mathrm{loc}}(\mathbb{R}^{2d})$
		for any $0<r<s-2d$, hence uniformly on compact subsets.
	\end{enumerate}
\end{theorem}

The previous convergence result is expected to improve in the smooth context in view of the characterization given in \eqref{S000 char}.
\begin{corollary}[{\cite[Cor. 1.2]{nt pointw}}]
	\label{maint s000} Let  $H_{0}=a^{\mathrm{w}}$
	as above and $V\in C^{\infty}_b(\rd)$. Let $\cA_t$ denote the classical flow associated with $H_0$ as in \eqref{block At}. For any $t\in \rte$:
	\begin{enumerate}
		\item the distributions $e^{-2\pi i\Phi_{t}}e_{n,t}$, $n\geq 1$,
		and $e^{-2\pi i\Phi_{t}}u_t$
		belong to a bounded subset of $C^{\infty}_b(\mathbb{R}^{2d})$;
		\item $e_{n,t}\rightarrow u_t$ in $C^{\infty}(\rdd)$,
		hence uniformly on compact  subsets together with any derivatives.
	\end{enumerate}
\end{corollary}

It is interesting to compare this result with those obtained by Fujiwara in \cite{fujiwara2 duke}, where convergence at the level of kernels in $C^{\infty}_b$-sense for short times was proved. In spite of different assumptions and approximation schemes, we stress that our result is global in time.

We conclude with a convergence result in the same spirit, for potentials in the Sj\"ostrand class.

\begin{theorem}[{\cite[Thm. 1.3]{nt pointw}}]
	\label{maint sjo} Let  $H_{0}=a^{\mathrm{w}}$
	as discussed above and $V\in M^{\infty,1}(\rd)$. Let $\cA_t$ denote the classical flow associated with $H_0$ as in \eqref{block At}. For any $t\in \rte$:
	\begin{enumerate}
		\item the distributions $e^{-2\pi i\Phi_{t}}e_{n,t}$, $n\geq 1$,
		and $e^{-2\pi i\Phi_{t}}u_t$
		belong to a bounded subset of $M^{\infty,1}(\mathbb{R}^{2d})$;
		\item $e_{n,t}\rightarrow u_t$ in $(\mathcal{F}L^{1})_{\mathrm{loc}}(\rdd)$,
		hence uniformly on compact subsets.
	\end{enumerate}
\end{theorem}

We stress that a typical potential setting in the papers by Albeverio and coauthors is ``harmonic oscillator plus a bounded perturbation'', the latter in the form of the Fourier transform of a (finite) complex measure on $\rd$ - cf. \cite{albeverio book} and the references therein. While those results rely on completely different techniques (in particular, infinite-dimensional oscillatory integral operators), in view of the embedding  $\cF\cM(\rd) \subset M^{\infty,1}(\rd)$ proved in \cite[Prop. 3.4]{nt} we are able to cover this class of potentials too.

In addition to the regularity properties mentioned insofar, our choice of modulation space is particularly well suited to the problem in view of the rich algebraic structure discussed in Section \ref{sec prel}. The key of the proofs is that for $t\in \rte$ the approximate operator $E_n(t)$ can be expressed in integral form and a manageable form of the kernel $e_{n,t}$ can be derived. In particular, with the help of some technical lemmas we are able to write
\begin{align}
E_{n}\left(t\right)\varphi(x) & =a_{n,t}^\w\,\mu\left(\mathcal{A}_{t}\right)\varphi (x) \nonumber \\
& =c(t)\left| \det B_{t} \right|^{-1/2}\int_{\mathbb{R}^{d}}e^{2\pi i\Phi_{t}\left(x,y\right)}\widetilde{a_{n,t}}\left(x,y\right)\varphi \left(y\right)dy,\label{eq ent form}
\end{align}
where $\Phi_t$ is as in \eqref{phit} and $\{a_{n,t}\},\{\widetilde{a_{n,t}}\} \subset M^{\infty}_s(\rdd)$ are bounded sequences of symbols for fixed $t\in \rte$.

\section{Results on integral kernels at exceptional times}
The occurrence of a set of exceptional times in Theorems \ref{maint minfty} and \ref{maint sjo} comes not as a surprise from a mathematical point of view: it may happen indeed that the integral kernel of the evolution operator degenerates into a distribution. A standard example of this phenomenon is provided by the harmonic oscillator, namely
\[ i\partial_t \psi = -\frac{1}{4\pi}\triangle \psi + \pi|x|^2 \psi. \] The integral kernel of the corresponding evolution operator is known as the \textit{Mehler kernel} and can be explicitly characterized \cite{dG symp met,kapit}: for $k \in \bZ$,
\begin{equation}\label{mehler} u_t(x,y) = \begin{cases} c(k) |\sin t|^{-d/2}\exp \left( \pi i \frac{x^2+y^2}{\tan t} - 2\pi i \frac{x\cdot y}{\sin t} \right) & (\pi k < t < \pi(k+1)) \\ c'(k)\delta((-1)^k x-y) & (t=k\pi) \end{cases}, \end{equation} for suitable phase factors $c(k),c'(k)\in \bC$. This shows the expected degenerate behaviour at integer multiples of $\pi$, which is consistent with the fact that the associated classical flow $\cA_t$ is given by
\[ \cA_t = \left( \begin{array}{cc} (\cos t) I & (\sin t) I \\ -(\sin t)I & (\cos t)I \end{array} \right), \] where $I \in \bR^{d\times d}$ is the identity matrix. Hence we retrieve $\te= \{ t\in \bR \,:\, \sin t = 0\} = \{k\pi:\ \, k \in \bZ\}$.

We may wonder whether convergence of integral kernels still occurs in some distributional sense, hopefully better than the broadest one (that is $\cS'(\rdd)$). In view of the discussion in Section \ref{sec prel} on the triple $(M^1,L^2,M^\infty)$, a suitable setting may be provided by $M^{\infty}$. We have indeed a general result for the kernels of strongly convergent sequences of operators in $L^2$.

\begin{theorem}\label{thm conv minfty}
	Let $\{A_n\}\subset \cB(L^2(\rd))$, $n \in \bN$, be a sequence of bounded linear operators on $L^2(\rd)$ with associated distribution kernels $\{ a_n\} \subset \cS'(\rdd)$, and $A \in \cB(L^2(\rd))$ with distribution kernel $a \in \cS'(\rdd)$. Assume that $A_n \to A$ in the strong operator topology. Then:
	\begin{enumerate}
		\item $a_n,a \in M^{\infty}(\rdd)$, $n \in \bN$;
		\item $a_n \to a$ in the weak-* topology on $M^\infty(\rdd)$.
	\end{enumerate} In particular we have $a_n \to a$ in $\cF L^\infty_{\mathrm{loc}}(\rdd)$, the latter space endowed with the topology $\sigma((\cF L^\infty)_{\mathrm{loc}}(\rdd),(\cF L^1)_{\mathrm{comp}}(\rdd))$.
\end{theorem}

\begin{proof}
We have that $\{A_n\}$ is a bounded sequence in $\cB(L^2(\rd))$ as a consequence of the uniform boundedness principle, hence also in $\cB(M^1(\rd),M^\infty(\rd))$. The Feichtinger kernel theorem (Theorem \ref{fei ker thm}) yields that the kernels $a_n$ belong to a bounded subset of $M^{\infty}(\rdd)$. Similarly, $A \in \cB(L^2(\rd)) \Rightarrow a \in M^\infty(\rdd)$. For the second part of the claim we remark that $A_n \to A$ in the strong operator topology implies that $a_n \to a$ in $\cS'(\rdd)$. Therefore, for any fixed non-zero $g \in \cS(\rd)$ we have $V_ga_n \to V_g a$ pointwise in $\rdd$. Moreover, we have the estimate $|V_g a_n(x,\xi)|\le C$, for some constant $C>0$ independent of $n$ by the first part of the claim. Hence, for any $\varphi \in M^1(\rd)$ we have
\begin{align*}
\la a_n,\varphi \ra & = \int_{\rdd}V_ga_n(x,\xi) \overline{V_g\varphi(x,\xi)} dxd\xi \\ & \to \int_{\rdd}V_ga(x,\xi) \overline{V_g\varphi(x,\xi)} dxd\xi = \la a,\varphi\ra, \end{align*} by the dominated convergence theorem.
\end{proof}
It would be interesting to prove the  boundedness of $a_n$ in $M^\infty(\rdd)$ in Theorem \ref{thm conv minfty} without using the uniform boundedness principle, although it could be not immediate. 

A straightforward application of this result allows us to prove global-in-time convergence of integral kernels, although in a weaker sense than before.

\begin{corollary}\label{cor conv minfty}
	Assume $V \in L^{\infty}(\rd)$. Let $e_{n,t}\in \cS'(\rdd)$ be the distribution kernel of the Feynman-Trotter parametrix $E_n(t)$ in \eqref{trotter maint} and $u_t \in \cS'(\rdd)$ be the kernel of the Schr\"odinger evolution operator $U(t)$ associated with the Cauchy problem \eqref{cauchy schr}. For any $n \in \bN$ and $t \in \bR$ we have $e_{n,t}, u \in M^\infty(\rdd)$. Moreover, $e_{n,t}\to u_t$  in the weak-* topology on $M^{\infty}(\rdd)$ for any fixed $t \in \bR$.
\end{corollary}

For more regular potentials we expect that the conclusion of Corollary \ref{cor conv minfty} can be improved. Let us first provide a version of the Trotter formula for potentials in $ M^{\infty,1}(\rd)$, with strong convergence on $M^1(\rd)$.

\begin{theorem}
	Assume $V \in M^{\infty,1}(\rd)$. Let $\{ E_n(t)\}$ be the sequence of Feynman-Trotter parametrices defined in \eqref{trotter maint} and $U(t)$ be the Schr\"odinger evolution operator $U(t)$ associated with the Cauchy problem \eqref{cauchy schr}. For any fixed $t \in \bR$ we have
	\[ \lim_{n\rightarrow\infty} E_n(t) = U(t), \quad \lim_{n\rightarrow\infty} E_n(t)^* = U(t)^*  \] in the strong topology of operators acting on $M^1(\rd)$.
\end{theorem}

\begin{proof}
	We prove that $E_n(t)\to U(t)$ strongly in $\cB(M^1(\rd))$; the claim concerning adjoint operators follows by similar arguments since $U(t)^* =U(-t)$ and $E_n(t)^*=\lc e^{i\frac{t}{n}V}e^{i\frac{t}{n}H_0} \rc^n$. \par
As already observed, we know that the operator $-iH_0$ with domain $D(H_0)= \{ \varphi \in L^2(\rd) \,:\, H_0\varphi \in L^2(\rd) \}$ is self-adjoint (\cite{hormander1 mehler}). Let $U_0(t)=e^{-itH_0}$ be the corresponding strongly   continuous unitary group on $L^2(\rd)$. The well-posedness of the Schr\"odinger equation $i\partial_t \psi = H_0 \psi$ in $M^1(\rd)$ (see e.g. \cite{CNR rough}) implies that the restriction of $U_0(t)$ to $M^1(\rd)$ defines a strongly continuous group on $M^1(\rd)$, its generator being the restriction of $H_0$ to the subspace $\{ \varphi \in M^1(\rd) \,:\, H_0\varphi \in M^1(\rd) \}$, as a consequence of known results on subspace semigroups, cf. \cite[Chapter 2, Sec. 2.3]{engel}. Since the pointwise multiplication by $V\in M^{\infty,1}(\rd)$ defines a bounded operator on $M^1(\rd)$, the desired result follows from the classical Trotter formula (\cite[Cor. 2.7 and Ex. 2.9]{engel}).
\end{proof}	

We provide an equivalent formulation of the previous result for the corresponding integral kernels, which is indeed a partial counterpart of the pointwise convergence results of Section \ref{sec pw conv}.

\begin{theorem}
	Under the same assumptions of Theorem \ref{maint sjo}, for all $t \in \bR$ and $\varphi \in M^1(\rd)$, the functions
	\[ \la e_{n,t}(x,\cdot),\varphi\ra, \quad \la e_{n,t}(\cdot,y),\varphi\ra, \quad\la u_t(x,\cdot),\varphi\ra, \quad\la u_t(\cdot,y),\varphi\ra \]
	belong to $M^1(\rd)$.\par
	 Moreover
	\[ \la e_{n,t}(x,\cdot),\varphi\ra \to \la u_t(x,\cdot),\varphi\ra, \quad \la e_{n,t}(\cdot,y),\varphi\ra \to \la u_t(\cdot,y),\varphi\ra \]
in $M^1(\rd)$, hence in $L^p(\rd)$ for every $1\leq p\leq \infty$.
\end{theorem}
The last conclusion follows from the continuous embedding $M^1(\rd)\hookrightarrow L^p(\rd)$, for every $1\leq p\leq \infty$.
\begin{remark} We expect other improvements of Theorem \ref{thm conv minfty} to hold in the case where $A_n = E_n(t)$, $A=U(t)$. In particular, convergence result for the corresponding integral kernels could be investigated in the context of mixed modulation spaces and generalized kernel theorems in the spirit of \cite{cn ker}. We will not engage in such formulation here in order to avoid quite technical discussions.
\end{remark}

\section{Physics at exceptional times}
In spite of the attempts to shed light on the nature of exceptional times and the partial results in the previous section, a physical interpretation of exceptional times is still not clear at the moment. This non-trivial question also appears in the form of an enigmatic exercise in the textbook \cite[Problem 3-1]{feyn2 hibbs} by Feynman and Hibbs. While dimensional analysis and heuristic arguments may provide some hints, a precise answer still seems to be missing.

We give our contribution to this discussion with a short argument which elucidates the nature of exceptional times in terms of measurable quantities. Recall that $B(u,r)$ denotes the ball with center $u \in \rd$ and radius $r>0$ in $\rd$. Following the custom in physics we adopt below the bra-ket notation, and we identify states with their wave functions in the position representation.

Fix $x_0,y_0 \in \rd$ and $a,b>0$, and consider the normalised wave-packets
\[ |A\ra = \frac{1}{\sqrt{|B(y_0,a)|}} \mathbbm{1}_{B(y_0,a)}
, \quad |B\ra = \frac{1}{\sqrt{|B(x_0,b)|}} \mathbbm{1}_{B(x_0,b)}.  \]
The corresponding transition amplitude from the state $|A\ra$ to $|B\ra$ under the Hamiltonian $H=H_0 + V$ as in Theorem \ref{maint sjo}, namely
\[ I=I(t,x_0,y_0,a,b) = \la B | U(t) | A\ra, \quad t \in \bR, \]
trivially satisfies the estimate
\[ |I(t,x_0,y_0,a,b)| \le 1, \quad \forall t \in \bR,\, x_0,y_0 \in \rd, \, a,b>0. \]
This bound cannot be improved at exceptional times: consider for instance the case where $t=0$, $x_0 = y_0$ and $a=b$, which yields $I=1$. Nevertheless, we have the following result.

\begin{proposition}
	Under the same assumptions of Theorem \ref{maint sjo}, for all $t \in \rte$ and $x_0,y_0 \in \rd$ we have
	\[ \lim_{a,b\to0} \frac{I(t,x_0,y_0,a,b)}{(ab)^{d/2}} = C\overline{u_t(x_0,y_0)}, \] where $C=C(d)=|B(0,1)|$.
\end{proposition}

\begin{proof} An explicit computation yields
	\[ \frac{I(t,x_0,y_0,a,b)}{C(ab)^{d/2}} = \frac{1}{C^2(ab)^{d}} \int_{B(x_0,b)} \int_{B(y_0,a)} \overline{u_t(x,y)}dydx, \]
	and the conclusion follows by the continuity of $u_t(x,y)$ in $\rdd$, because $u_t \in \big(\mathcal{F} L^1\big)_{\rm loc}(\rdd)$ for $t \in \rte$ by Theorem \ref{maint sjo}.
\end{proof}

This result shows that while $|I| \le 1$ in general, for a non-exceptional time $t \in \rte$ we have that $|I| \sim (ab)^{d/2}$  as $a,b\to 0$. In particular $|I| \to 0$ as $a,b\to 0$ except (possibly) for exceptional times.

\section*{Acknowledgements}
We would like to express our gratitude to Professors Elena Cordero, Ernesto De Vito and Stephan Waldmann for fruitful conversations on the topics of this paper.

\end{document}